\documentclass[letterpaper, 10 pt, conference]{ieeeconf}
\usepackage{latexsym,tabularx}
\usepackage{graphicx}
\usepackage{epstopdf}
\graphicspath{{./figure/}}
\DeclareGraphicsExtensions{.pdf,.jpeg,.jpg,.png,.eps}
\usepackage{array}
\usepackage{epsfig}
\usepackage{subfig}
\usepackage{psfrag,psfragx}
\usepackage{times}	
\usepackage{tgtermes}	
\usepackage[T1]{fontenc}	
\usepackage{yfonts}
\usepackage{pgothic}
\usepackage{adjustbox}
\usepackage{amsmath,amsfonts,amssymb,array,tabulary}
\usepackage{upgreek}
\usepackage{float}
\usepackage{arydshln}
\usepackage{algorithm,algorithmic}
\usepackage{color,url}
\usepackage{framed}
\usepackage{mathdots}
\usepackage{mathrsfs} 
\usepackage{tikz}
\usetikzlibrary{decorations.pathmorphing}
\usetikzlibrary{arrows,automata,positioning,shapes}
\usetikzlibrary{decorations.pathmorphing}
\tikzset{snake it/.style={decorate, decoration=snake}}
\usepackage{pdflscape}
\usepackage{mathtools} \mathtoolsset{showonlyrefs}
\usepackage{theoremref} 
\usepackage{cite}
\newtheorem{theorem}{Theorem}
\newtheorem{lemma}{Lemma}
\newtheorem{corollary}{Corollary}

\newtheorem{definition}{Definition}

\newcommand{\nd}[1][d]{\mathcal{N}_{#1}}
\newcommand{\pd}[1][d]{\mathcal{P}_{#1}}

\usepackage{filecontents}

\IEEEoverridecommandlockouts
\title{Herdable Systems Over Signed, Directed Graphs}
\author{ Sebastian F. Ruf$^{1}$, Magnus Egerstedt$^{1}$, and Jeff S. Shamma$^{2}$
	\thanks{$^{1}$S.F. Ruf and M. Egerstedt are with the School of Electrical and Computer Engineering, Georgia Institute of Technology and can be reached at
		{\tt\small ruf@gatech.edu} and {\tt\small magnus@gatech.edu} resp.}%
	\thanks{$^{2}$ J.S.Shamma is with Computer, Electrical and Mathematical Science and Engineering Division, King Abdullah University of Science and Technology, Saudi Arabia, {\tt\small jeff.shamma@kaust.edu.sa}. The research reported in this publication was supported by funding from King Abdullah University of Science and Technology (KAUST).}%
}

\begin{document}
	\maketitle
	\thispagestyle{empty}
	\pagestyle{empty}
	\begin{abstract}
		This paper considers the notion of herdability, a set-based reachability condition, which asks whether the state of a system can be controlled to be element-wise larger than a non-negative threshold. The basic theory of herdable systems is presented, including a necessary and sufficient condition for herdability. This paper then considers the impact of the underlying graph structure of a linear system on the herdability of the system, for the case where the graph is represented as signed and directed. By classifying nodes based on the length and sign of walks from an input, we find a class of completely herdable systems as well as provide a complete characterization of nodes that can be herded in systems with an underlying graph that is a directed out-branching rooted at a single input.
	\end{abstract}
\section{Introduction}\label{sec:intro}

Controllability is a fundamental property of a dynamical system, and has been an area of study since the work of Kalman et. al in the 1960s \cite{kalmancont}. However there are cases where a system need not be completely controllable to achieve desirable system outcomes. This paper considers a class of these systems by considering the reachability of a specific set, rather than the whole state space as in controllability. 

As an example, consider the case where the state of a dynamical system represents the percentage of a given community that will vote for a political candidate and the control input represents advertising. Here an advertising campaign is successful if the state can be driven high enough for the candidate to win, regardless of whether communities can be made to vote at any specific percentage as would be required by complete controllability. 

In order to study systems that are not completely controllable but for which certain desirable control outcomes are still achievable, this paper introduces a set-based reachability condition known as herdability, which considers whether the components of the state can be driven above a non-negative threshold. This target set describes desired behavior in social and biological sciences where many systems act based on thresholds, for example collective social behavior \cite{granovetter1978threshold} and the firing of a neuron \cite{hodgkin1952quantitative}. More formally, a continuous time, linear system,\begin{equation}\label{eq:sys}
\dot{\mathbf{x}}=A\mathbf{x}+B\mathbf{u}
\end{equation}
where $A\in\mathbb{R}^{n\times n}$ and $B\in\mathbb{R}^{n\times m}$, is \emph{completely herdable} if there exists a control input that makes the state enter the set $\mathcal{H}_d = \{ \mathbf{x} \in \mathbb R^{n} : \mathbf{x}_{i} \geq d\}$ for all $d\geq0$, where $ \mathbf{x}_{i}$ is the $i$-th element of $\mathbf{x}$. 
 Returning to the example of voting in an election, $\mathbf{x}_{i}$ is now the percentage of community $i$ that will vote for a candidate and reaching the set $\mathcal{H}_{.5} = \{ \mathbf{x} \in \mathbb R^{n} : (\mathbf{x})_{i} \geq .5\}$ wins the election.

This paper considers the herdability of linear systems based on the structure of the underlying interaction graph, which encodes information about how states and inputs interact with each other. The relationship between a graph and a dynamical system has been previously considered using two primary approaches. 

The first approach moves from a specific graph structure to a system dynamic, most often consensus dynamics \cite{mesbahi2010graph}. The controllability of these consensus system has been shown to be directly related to the structure of the graph, in that system  controllability is lost when nodes are symmetric with respect to an input \cite{ rahmani2009controllability,martini2010controllability,zhang2014upper,yaziciouglu2016graph,chapman2014symmetry,alemzadeh2017controllability}. As will be seen in this paper, it is in fact a loss of symmetry that causes a loss of herdability. 

The second approach takes a dynamical system and maps it to a graph to discuss properties of all systems that share the same graph structure. This approach is known as structural controllability \cite{Lin1974,shields1976structural,glover1976characterization}.
In structural controllability, a dynamical system is represented by a graph in which each edge of the graph is assigned a weight in $\mathbb{R}$. A system is structurally controllable if and only if it is controllable for almost all weights that are assigned to the edges, which can be verified directly from the structure of the underlying graph. Structural controllability has been extended to sign controllability, which determines controllability based on the sign pattern of the system matrices \cite{johnson1993sign,hartung2013characterization} or the sign pattern of the underlying graph \cite{tsatsomeros1998sign}.
 
This paper shares the approach of sign controllability in that the control properties of classes of systems are considered based on the sign pattern of the graph structure. Specifically this paper represents the interaction structure as a signed, directed graph as the sign pattern is often sufficient information to determine the herdability of a system. 
 Signed graphs are used in the social networks context to represent systems in which agents are both friends and enemies \cite{easley2010networks}. In this light, the central problem of the paper can be phrased in a social networks context as follows: how does the grouping of friends and enemies in the network relate to the ability to convince agents in the system to hold an opinion in $\mathcal{H}_{d}$?

The rest of the paper is organized as follows: Section~\ref{sec:herd} introduces the basic theory of herdable system. In Section~\ref{sec:pre} a graph theoretic characterization of the interaction structure of a linear system is presented. Section \ref{sec:nes} considers a necessary condition for complete herdability based on the underlying graph structure. Section \ref{sec:class} considers the sign herdability of a system based on the underlying graph. The paper concludes in Section~\ref{sec:conc}.
	 
	 	 \subsection*{Notation:} For a vector $\mathbf{k}\in\mathbb{R}^{n}$, $\mathbf{k}_{i}$ refers to the $i$-th element of $\mathbf{k}$. For a matrix $K\in\mathbb{R}^{n\times m}$, $K_{i,:}$ refers to the $i$-th row of $K$, $K_{:,j}$ refers to the $j$-th column of $K$ and $K_{i,j}$ to the $i,j$-th element of $K$. The cardinality of the set $\mathcal{S}$ is expressed as $|S|$. Let $\mathrm{sgn}(\cdot)$ denote the sign function which follows
	 	 \begin{equation}
	 	 \mathrm{sgn}(x) = \left\{\begin{array}{lr}
	 	 -1 & \text{for } x<0,\\
	 	 0 & \text{for } x=0,\\
	 	 1 & \text{for } x>0.
	 	 \end{array}\right.
	 	 \end{equation}  Let $\mathbf{0}_{n}\in\mathbb{R}^{n}$ be a vector of zeros, $\mathbf{1}_{n}\in\mathbb{R}^{n}$ be a vector of ones.
	 	 Logical AND is denoted by $\wedge$ and $\veebar$ denotes logical EXCLUSIVE OR.  
	 	 A vector is unisigned if all non-zero elements have the same sign. If a vector $\mathbf{v}$ is unisigned then $\mathrm{sgn}(\mathbf{v})=1$ ($\mathrm{sgn}(\mathbf{v})=-1$) is all elements of $\mathbf{v}$ are positive (negative). 
	 	 \section{Characterizing Herdability}\label{sec:herd}
	 	 In this section, the basic theory of the herdability of continuous time, linear dynamical systems is presented as well as a characterization of herdability based on the system controllability matrix. Of course before characterizing herdability, the following definitions of herdability are required.  
	 	 \begin{definition}
	 	 	The state $i$ of a linear system is herdable if $\forall \mathbf{x}(0)\in\mathbb{R}^{n}$ and $h\geq 0$, there exists a finite time $t_{f}$ and an input $\mathbf{u}(t), \ t\in[0,t_{f}]$ such that $\mathbf{x}(t_{f})_{i}\geq h$ under $\mathbf{u}(t)$. 
	 	 \end{definition}
	 	 \begin{definition}
	 	 	A set of states, $\mathcal{X}\subseteq\{1,2,\dots,n\}$, is herdable if $\forall \mathbf{x}(0)\in\mathbb{R}^{n}$ and $h\geq 0$, there exists a finite time $t_{f}$ and an input $\mathbf{u}(t), \ t\in[0,t_{f}]$ such that $\mathbf{x}(t_{f})_{i}\geq h, \forall i\in \mathcal{X}$ under control input $\mathbf{u}(t)$. 
	 	 \end{definition}
	 	 \begin{definition}
	 	 	A linear system is completely herdable if the set of states $\mathcal{X}=\{1,2,\dots,n\}$ is herdable.
	 	 \end{definition} 
	 	 To translate the definition of herdability to a necessary and sufficient condition for herdability requires some basic concepts from the study of linear systems, specifically the interplay between the reachable subspace and the controllability matrix.  
	
	 	 Define the reachable subspace $\mathcal{R}[0,t]$ as
	 	 \begin{equation}
	 	 \mathcal{R}[0,t]=\left\{\mathbf{x}_{1}\in\mathbb{R}^{n}:\exists \mathbf{u}(\cdot), \mathbf{x}_{1}=\int_{0}^{t}e^{A(t-\tau)}B\mathbf{u}(\tau)d\tau\right\}.
	 	 \end{equation}   
	 	 
	 	 \noindent The controllability matrix $\mathcal{C}$ of a linear system is 
	 	  \begin{equation}
	 	 \mathcal{C}=\left[B,AB,A^{2}B,\dots,A^{n-1}B \right]
	 	 \end{equation}
	 	 
	 	 It is possible to characterize the herdability of a system based on its controllability matrix. Recall the following from \cite{hespanha2009linear} (though any introductory linear systems text will do):
	 	 \begin{lemma}Theorem 11.5 from \thlabel{lem:rceq} \cite{hespanha2009linear} \begin{equation}
	 	 	\mathcal{R}[0,t]=\mathrm{range}(\mathcal{C}).
	 	 	\end{equation} 
	 	 \end{lemma}
	 	 
	 	 With \thref{lem:rceq} it is possible to prove the following condition for the herdability of a set of states.  
	 	 \begin{theorem}\thlabel{lem:herdset}
	 	 	A set of states $\mathcal{X}\subseteq\{1,2,\dots,n\}$ is herdable if and only if there is exists a vector $\mathbf{k}\in \mathrm{range}(\mathcal{C})$ that satisfies $\mathbf{k}_{i}>0$ for all $i\in\mathcal{X}$.
	 	 \end{theorem}
	 	 \begin{proof}
	 	 	Define $\mathcal{K}$ to be the set that contains the positive elements of $\mathbf{k}$, $\mathcal{K}=\{p\in\mathbb{R} \ | \ p>0 \ \wedge \ \exists \ i \text{ such that } \mathbf{k}_{i}=p\}.$ 
	 	 	
	 	 	($\mathbf{k}\in \mathrm{range}(\mathcal{C})\Rightarrow$ $\mathcal{X}$ is herdable) Consider the problem of controlling all states in the set $\mathcal{X}$ to be greater than some lower threshold $h\geq 0$ from an initial condition $\mathbf{x}(0)$. Suppose there is a $\mathbf{k}\in \mathrm{range}(\mathcal{C})$, that satisfies $\mathbf{k}_{i}>0$ if $i\in\mathcal{X}$. As $\mathbf{k}\in \mathrm{range}(\mathcal{C})$, $\exists \pmb{\alpha}$ such that 
	 	 	$\mathcal{C}\pmb{\alpha}=\mathbf{k}.$
	 	 	If \begin{equation}\gamma>\frac{\max_{j} \  (h\mathbf{1}_{n}-e^{At}\mathbf{x}(0))_{j}}{\min{\mathcal{K}}}\end{equation} and $\mathbf{v}=\gamma\pmb{\alpha}$ then for all $i\in \mathcal{X}$ it holds that
	 	 	\begin{equation}
	 	 	(\mathcal{C}\mathbf{v})_{i} > (h\mathbf{1}_{n}-e^{At}\mathbf{x}(0))_{i}.
	 	 	\end{equation} 
	 	 	As the range of $\mathcal{C}$ is the same as the reachable subspace, $\exists \mathbf{u}(\cdot)$ such that for all $i\in \mathcal{X}$ 
	 	 	\begin{equation}(e^{At}\mathbf{x}(0)+\int_{0}^{t}e^{A(t-\tau)}B\mathbf{u}(\tau)d\tau)_{i}>h
	 	 	\end{equation}
	 	 	then all states in $\mathcal{X}$ can be made larger that $h$ and as $h$ is arbitrary the subset of states $\mathcal{X}$ is herdable.
	 	 	
	 	 	($\mathcal{X}$ is herdable $\Rightarrow\mathbf{k}\in \mathrm{range}(\mathcal{C})$) 
	 	 	As the set of state nodes $\mathcal{X}$ is herdable, each element of $\mathcal{X}$ can be made larger than some $h^{\ast}>0$ from any initial condition. Consider the initial condition $x(0)=\mathbf{0}_{n}$. Then by the herdability of the set $\mathcal{X}$ there exists a vector $\mathbf{k}^{\ast}$ that satisfies $\mathbf{k}^{\ast}_{i}>h^{\ast} \ \forall i\in \mathcal{X}$ and an input $\mathbf{u}(\cdot)$ such that \begin{equation}\int_{0}^{t}e^{A(t-\tau)}B\mathbf{u}(\tau)d\tau=\mathbf{k}^{\ast}
	 	 	\end{equation} Then $\mathbf{k}^{\ast}_{i}>h^{\ast}>0 \ \forall i\in\mathcal{X}$. By the definition of $\mathcal{R}[0,t]$, $\mathbf{k}^{\ast}\in \mathcal{R}[0,t]$ and consequently $\mathbf{k}^{\ast}\in \mathrm{range}(\mathcal{C})$ by \thref{lem:rceq}.
	 	 \end{proof}
	 	 
	 	 \begin{corollary}\thlabel{cor:krange}
	 	 	A linear system is completely herdable if and only if there exists an element-wise positive vector $\mathbf{k}\in \mathrm{range}(\mathcal{C})$. 
	 	 \end{corollary}
	 	 
	 	 Verifying the existence of a element-wise positive vector $\mathbf{k}\in \mathrm{range}(\mathcal{C})$ can be difficult, especially when dealing with large systems. As such, this paper considers a sufficient condition for complete herdability, which will subsequently be expanded on based on the underlying system graph.
	 	 \begin{theorem}\thlabel{th:ccond}
	 	 	If for each state $i \in \{1,2,\dots,n\}$, there exists a $j$ such that $(\mathcal{C})_{i,j}\neq 0$ and the vector $(\mathcal{C})_{:,j}$ is unisigned, then the system is completely herdable. 
	 	 \end{theorem}
	 	 \begin{proof}
	 	 	For state $i$ let $z^{i}$ be the column which corresponds to the unisigned vector with non-zero element $i$. Consider $\Gamma=\bigcup_{i} z^{i}$, the set of all $z^{i}$ such that there is no repeated values. This is necessary as it may be that for two states $i$ and $j$, $z^{i}=z^{j}$. 
	 	 	
	 	 	Construct the vector $\pmb{\alpha}$ such that $\pmb{\alpha}_\kappa=0$ if $\kappa\notin\Gamma$, $\pmb{\alpha}_\kappa=1$ if $\kappa\in\Gamma$ and $\mathrm{sgn}((\mathcal{C})_{:,\kappa})=1$ and $\pmb{\alpha}_\kappa=-1$ if $\kappa\in\Gamma$ and $\mathrm{sgn}((\mathcal{C})_{:,\kappa})=-1$. As the condition of the Theorem holds for all $i \in \{1,2,\dots,n\}$, there exists $ \mathbf{k}\in\mathbb{R}^{n}$ which is element wise positive such that
	 	 	$\mathcal{C}\pmb{\alpha}=\mathbf{k}$ and the system is completely herdable by \thref{cor:krange}.
	 	 \end{proof}
	 	 \section{Characterizing Dynamical Systems via Graphs}\label{sec:pre}
	 	 
	 This paper concerns itself with the underlying graphical structure of the dynamical system in Equation~\eqref{eq:sys}; more specifically its representation as one of two graphs. Each of these graphs contain different levels of information about the interactions between the states and inputs. The first graph is a signed graph $G^{s}=(\mathcal{V},\mathcal{E},s(\cdot))$ where $s(\cdot)$ accepts an edge and returns a label in $\{+1,-1\}$, which is the sign of the edge. This signed graph represents a class of systems whose edge weights all have the same sign pattern.   
	 The second graph is a weighted graph $G^{w}=(\mathcal{V},\mathcal{E},w(\cdot))$ where $w(\cdot)$ accepts an edge and returns a weight in $\mathbb{R}$. The weighted graph is the representation of a single system. 
	 
	 As will be shown, the weighted graph $G^{w}$ can be directly related to the controllability matrix and therefore the controllability properties of the system. This paper focuses on the interplay between $G^{s}$ and $G^{w}$, in that the presented structural results are cases where the results for the herdability of a system based on the weighted $G^{w}$ can be extended to all signed graphs with the same sign structure $G^{s}$ regardless of the weights of the edges in $G^{w}$, a notion similar to sign controllability \cite{johnson1993sign,tsatsomeros1998sign,hartung2013characterization}. As such, this paper will concern itself with extending \thref{th:ccond} to consider the sign herdability of a class of systems. 
	 
	 The formal definition of the graphs follows. The set of vertices satisfies $\mathcal{V}=\mathcal{V}_{x}\cup\mathcal{V}_{u}, \ \mathcal{V}_{x}\cap\mathcal{V}_{u}=\emptyset,$ where $\mathcal{V}_{x}=\{v_{x1},v_{x2},\dots,v_{xn}\}$ is a set of vertices representing the states of the system and $\mathcal{V}_{u}=\{v_{u1},v_{u2},\dots,v_{um}\}$ represents the inputs to the system. An arbitrary element of $\mathcal{V}$ will be referred to by $v_{i}$ for some index $i$, as will elements $v_{xi}\in\mathcal{V}_{x}$ and $v_{ui}\in \mathcal{V}_{u}$. The state $i$ will be interchangeably referred to by the node $v_{xi}$ as will the input $j$ and the node $v_{uj}$.  
	 
	 The edge set satisfies $\mathcal{E}=\mathcal{E}_{x}\cup\mathcal{E}_{u}$ where the edges in $\mathcal{E}_{x}$ represent interactions between states of the system, while $\mathcal{E}_{u}$ represents interactions between the inputs and the states. Denote the directed edge from $v_{i}$ to $v_{j}$ as $(v_{i},v_{j})$. Then $(v_{xi},v_{xj})\in\mathcal{E}_{x} \Leftrightarrow A(j,i)\neq 0$ and $(v_{ui},v_{xj})\in\mathcal{E}_{u} \Leftrightarrow B(j,i)\neq 0$. An arbitrary element of $\mathcal{E}$ will be referred to by $e_{i}$ for some $i$. By partitioning the node and edges sets, it is possible to define the state subgraph $G_{x}=(\mathcal{V}_{x},\mathcal{E}_{x})$, which captures only interactions between states as well as the input subgraph $G_{u}=(\mathcal{V},\mathcal{E}_{u})$ which captures interactions from the inputs to the states. Note that the input nodes do not interact with each other nor is it possible to have an edge $(v_{xi},v_{uj})$. 

	 When considering the signed graph $G^{s}$,  $s((v_{xi},v_{xj})) =\mathrm{sgn}(A(j,i))$ and $s((v_{ui},v_{xj})) =\mathrm{sgn}(B(j,i))$. Similarly for $G^{w}$,   $w((v_{xi},v_{xj}))=A(j,i)$ and $w((v_{ui},v_{xj}))=B(j,i)$. 
	 
	 As an example, consider the system
	 \begin{equation}\label{eq:ex}
	 \dot{\mathbf{x}}=\begin{bmatrix}
	 -1 & 0 & 0\\
	 5 & 0 & 2 \\
	 4 & -3&  0
	 \end{bmatrix}+\begin{bmatrix} 0& -2\\
	 2& 0\\
	 0& 3
	 \end{bmatrix}\mathbf{u}
	 \end{equation} which is translated into $G^{s}$ and $G^{w}$ in Fig.~\ref{fig:ex}.
	 		 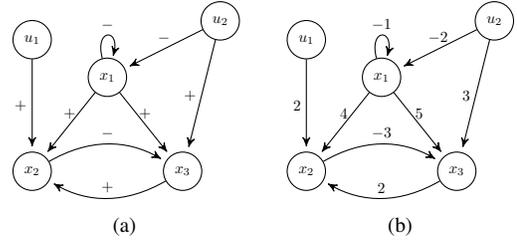
\begin{figure}[h]
	 		\centering
	 		\subfloat[]{\label{subfig:seq2}\begin{tikzpicture}[
	 			->,
	 			>=stealth',
	 			shorten >=2pt,
	 			auto,
	 			node distance=0.5cm,
	 			every text node part/.style={align=center},
	 			scale=0.25, every node/.style={scale=0.58}
	 			]
	 			\tikzstyle{every state}=[fill=none,draw=black,text=black]
	 			\node[
	 			state,
	 			] (n0) at(2,0)
	 			{$u_{2}$};
	 			\node[
	 			state,
	 			] (n1) at(-8,-1)
	 			{$u_{1}$};
	 			\node[
	 			state,
	 			] (n2) at(0,-8)
	 			{$x_{3}$};
	 			\node[
	 			state,
	 			] (n3) at(-8,-8)
	 			{$x_{2}$};
	 			\node[
	 			state,
	 			] (n4) at(-4,-3)
	 			{$x_{1}$};
	 			\path (n0) edge [->] node [left,above] {$-$} (n4)
	 			(n0) edge [->] node [left] {$+$} (n2)
	 			(n1) edge [->] node [left] {$+$} (n3)
	 			(n4) edge [->,loop above] node [left,above] {$-$} (n4)
	 			(n2) edge [->,bend left] node [left,above] {$+$} (n3)
	 			(n3) edge [->,bend left] node [left,above] {$-$} (n2)
	 			(n4) edge [->] node [left,above] {$+$}  (n2)
	 			(n4) edge [->] node [left,above] {$+$}  (n3);
	 			\end{tikzpicture}} 
 			\hspace{.5cm}
	 			 		\subfloat[]{\label{subfig:seq3}\begin{tikzpicture}[
	 			->,
	 			>=stealth',
	 			shorten >=2pt,
	 			auto,
	 			node distance=0.5cm,
	 			every text node part/.style={align=center},
	 			scale=0.25, every node/.style={scale=0.58}
	 			]
	 			\tikzstyle{every state}=[fill=none,draw=black,text=black]
	 			\node[
	 			state,
	 			] (n0) at(2,0)
	 			{$u_{2}$};
	 			\node[
	 			state,
	 			] (n1) at(-8,-1)
	 			{$u_{1}$};
	 			\node[
	 			state,
	 			] (n2) at(0,-8)
	 			{$x_{3}$};
	 			\node[
	 			state,
	 			] (n3) at(-8,-8)
	 			{$x_{2}$};
	 			\node[
	 			state,
	 			] (n4) at(-4,-3)
	 			{$x_{1}$};
	 			\path (n0) edge [->] node [left,above] {$-2$} (n4)
	 			(n0) edge [->] node [left] {$3$} (n2)
	 			(n1) edge [->] node [left] {$2$} (n3)
	 			(n4) edge [->,loop above] node [left,above] {$-1$} (n4)
	 			(n2) edge [->,bend left] node [left,above] {$2$} (n3)
	 			(n3) edge [->,bend left] node [left,above] {$-3$} (n2)
	 			(n4) edge [->] node [left,above] {$5$}  (n2)
	 			(n4) edge [->] node [left,above] {$4$}  (n3);
	 			\end{tikzpicture}}

	 		\caption{ The graphs of the system in in Equation \eqref{eq:ex}. 
	 			 ~\ref{subfig:seq2}: $G^{s}$ the signed graph. ~\ref{subfig:seq3}: $G^{w}$ the weighted graph.}
	 		 			\label{fig:ex}
	 	\end{figure}
 	
	 To describe these graphs requires a number of basic definitions from graph theory  \cite{harary2005structural}. A walk from $v_{0}$ to $v_{p}$, $\pi(v_{0},v_{p})$, is any alternating sequence of nodes and edges $\pi(v_{0},v_{p})=v_{0},e_{1},v_{1},e_{2},v_{2}\dots, v_{p-1},e_{p},v_{p}$ such that $v_{i} \in \mathcal{V} \ \forall i$ and $e_{i}=(v_{i-1},v_{i})\in \mathcal{E}$. The set of walks from $v_{0}$ to $v_{p}$ is $\theta(v_{0},v_{p})$.
	 A node $v_{j}$ is reachable from $v_{i}$, which will be written as $v_{i}\rightarrow v_{j}$, if $\theta(v_{i},v_{j})\neq\emptyset$\footnote{Reachability is discussed within both graph theory and control theory. This paper will use the term reachable in both senses, with clarification only if it is uncertain which notion of reachability is considered.}.  The length of a walk, $
	\mathrm{len}(\pi)$, is equal to the number of edges in $\pi$.
	
	A walk has an associated sign which follows \begin{equation}s(\pi)=\prod_{e_{i}\in\pi} \ s(e_{i}),\end{equation} as well as an associated weight: \begin{equation}w(\pi)=\prod_{e_{i}\in\pi} \  w(e_{i}).\end{equation} This is distinct from the weight of a walk as it is treated in many applications, such as shortest path algorithms, which consider $w(\pi)=\sum_{e_{i}\in\pi} \  w(e_{i})$  \cite{even2011graph}.
	 Referring back to the example in Fig.~\ref{fig:ex}, the walk $\pi(u_{1},x_{3})=u_{1},(u_{1},x_{2}),x_{2},(x_{2},x_{3}),x_{3}$  is of length $2$ and has $s(\pi(u_{1},x_{3}))=-1$ and $w(\pi(u_{1},x_{3}))=-6$.
	
To begin classifying the system in Equation \eqref{eq:sys} based on the signed graph $G^{s}$, we define two basic types of sets. 
	 Let $\nd^{j}$ be the set of nodes reachable from $v_{uj}$ via at least one negative walk of length $d$. Similarly $\pd^{j}$ is the set of nodes reachable from $v_{uj}$ through at least one positive walk of length $d$. If there is only one input to the system, the superscript will be dropped to refer to $\nd$ and $\pd$ instead of $\nd^{1}$ and $\pd^{1}$. Returning to the example shown in Fig. \ref{subfig:seq2}, $\mathcal{N}_{1}^{1}=\emptyset$,  $\mathcal{P}_{1}^{1}=\{x_{2}\}$, $\mathcal{N}_{2}^{1}=\{x_{3}\}$, and $\mathcal{P}_{2}^{1}=\emptyset$. 

 As will be seen, the sets $\pd^{j}$ and $\nd^{j}$ can provide sufficient information to determine the herdability of a system, and in doing so determine the sign herdability of a class of systems. To show this requires classifying the structure of the weighted graph $G^{w}$. Consider the total weight of positively signed walks from input $v_{uj}$ to node $v_{xi}$ with length $d$, \begin{equation}{\rho}_{j\rightarrow i,d}^{+}=\sum_{\pi\in \theta_{d}^{+}(v_{uj},v_{xi})}w(\pi),\end{equation}
 where $\theta_{d}^{+}(v_{uj},v_{xi})$ is the set of positive walks of length $d$ from $v_{uj}$ to $v_{xi}$. From the definition of $\pd^{j}$, it holds that ${\rho}_{j\rightarrow i,d}^{+}>0$ if $v_{xi}\in\pd^{j}$ and $0$ else. Similarly the total weight of negatively signed walks from input $v_{uj}$ to node $v_{xi}$ with length $d$ is \begin{equation}{\rho}_{j\rightarrow i,d}^{-}=\sum_{\pi\in \theta_{d}^{-}(v_{uj},v_{xi})}w(\pi),\end{equation} where $\theta_{d}^{-}(v_{uj},v_{xi})$ is the set of negative walks of length $d$ from $v_{uj}$ to $v_{xi}$ and it follows that ${\rho}_{j\rightarrow i,d}^{-}<0$ if $v_{xi}\in\nd^{j}$ and $0$ else. The weight of all walks from input $v_{uj}$ follows: 
  \begin{equation}{\rho}_{j\rightarrow i,d}={\rho}_{j\rightarrow i,d}^{+}+{\rho}_{j\rightarrow i,d}^{-}. \end{equation}
 Consider the example shown in Fig. \ref{fig:cancel}. 
  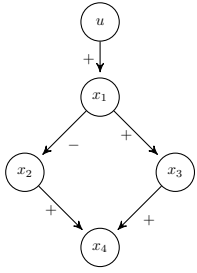
\begin{figure}[h]
 	\centering
 	\begin{tikzpicture}[
 	->,
 	>=stealth',
 	shorten >=2pt,
 	auto,
 	node distance=0.5cm,
 	every text node part/.style={align=center},
 	scale=0.25, every node/.style={scale=0.58}
 	]
 	\tikzstyle{every state}=[fill=none,draw=black,text=black]
 	\node[
 	state,
 	] (n0) at(0,0)
 	{$u$};
 	\node[
 	state,
 	] (n1) at(0,-4)
 	{$x_{1}$};
 	\node[
 	state,
 	] (n2) at(4,-8)
 	{$x_{3}$};
 	\node[
 	state,
 	] (n3) at(-4,-8)
 	{$x_{2}$};
 	\node[
 	state,
 	] (n4) at(0,-12)
 	{$x_{4}$};
 	\path (n0) edge [->] node [left] {$+$}(n1)
 	(n3) edge [->]  node [left] {$+$} (n4)
 	(n2) edge [->]  node {$+$} (n4)
 	(n1) edge [->] node [left] {$+$} (n2)
 	(n1) edge [->]  node {$-$} (n3);
 	\end{tikzpicture}
 	\caption{An example of a signed graph where $\nd$ and $\pd$ don't completely determine ${\rho}_{j\rightarrow i,d}$ }
 	\label{fig:cancel}
 \end{figure}

 The signed graph represents all systems of the form \begin{equation}
 \dot{\mathbf{x}}=\begin{bmatrix}
 0&0&0&0\\
 -\alpha_{1}&0&0&0\\
  \alpha_{2}&0&0&0\\
  0&\alpha_{3}&\alpha_{4}&0
 \end{bmatrix}\mathbf{x}+\begin{bmatrix}
 \beta_{1}\\
 0\\
 0\\
 0
 \end{bmatrix}\mathbf{u}
 \end{equation}
 where $\alpha_{1},\alpha_{2},\alpha_{3},\alpha_{4},\beta_{1}>0.$ Here the total walk weight to node $v_{x4}$ at length $2$ is 
 $\rho_{1\rightarrow4,2}=\beta_{1}\left(\alpha_{2}\alpha_{4}-\alpha_{1}\alpha_{3}\right),$ 
 the sign of which depends on the values of the various constants.
 
 The case where the sign of ${\rho}_{j\rightarrow i,d}$ is determined by $\nd^{j}$ and $\pd^{j}$ is shown in the following Lemmas. These Lemmas follow directly from the definitions of the sets $\pd^{j}$ and $\nd^{j}$ and as such are presented without proof.
\begin{lemma}\thlabel{lem:pdti}
	If $v_{xi}\in\pd^{j}\wedge v_{xi}\notin\nd^{j}$ then ${\rho}_{j\rightarrow i,d}>0$. 
\end{lemma}
\begin{lemma}\thlabel{lem:ndti}
	If $v_{xi}\in\nd^{j}\wedge v_{xi}\notin\pd^{j}$ then ${\rho}_{j\rightarrow i,d}<0$. 
\end{lemma} 

It is possible to relate ${\rho}_{j\rightarrow i,d}$ with the system matrices $A,B$ and ultimately the controllability properties of the system.  

 Define a weighted adjacency matrix $\tilde{A}_{w}$ for $G_{x}^{w}$, where $(\tilde{A}_{w})_{i,j}=w((v_{xj},v_{xi}))$ if $(v_{xj},v_{xi})\in\mathcal{E}_{x}$ and $(\tilde{A}_{w})_{i,j}=0$ if not. Define a weighted adjacency matrix $\tilde{B}_{w}$ for $G_{u}^{w}$, where $(\tilde{B}_{w})_{i,j}=w((v_{uj},v_{xi}))$ if $(v_{uj},v_{xi})\in\mathcal{E}_{u}$ and $(\tilde{B}_{w})_{i,j}=0$ if not. Note that from the definition of the weight of an edge, $\tilde{A}_{w}=A$ and $\tilde{B}_{w}=B$. 
 \begin{lemma}\thlabel{lem:matmult} \begin{equation}(A^{d-1}B)_{i,j}={\rho}_{j\rightarrow i,d}.\end{equation}
 	\end{lemma}
 \begin{proof}
 	The result will be shown via proof by induction on $d$. 
 	Consider the case of $d=1$. By the definition of the weight of an edge: 
 	$(B)_{i,j}={\rho}_{j\rightarrow i,1}.$ 
 	Consider the weight of all walks of length $d$ from an input $v_{uj}$ to a state node $v_{xi}$. By assumption, $(A^{d-2}B)_{i,j}={\rho}_{j\rightarrow i,d-1}$. Then 
 	\begin{equation}
 	(A^{d-1}B)_{i,j}=\sum_{k=1}^{n} (A)_{i,k}{\rho}_{j\rightarrow k,d-1}.
 	\end{equation} As a walk of length $d$ is the concatenation of a walk of length $d-1$ and a walk of length $1$, it follows from the definition of the weight of a walk that 
 	\begin{equation}
 	\sum_{k=1}^{n} (A)_{i,k}{\rho}_{j\rightarrow k,d-1}={\rho}_{j\rightarrow i,d}.
 	\end{equation}
 	\end{proof}
 
 	 	 Given the structure of $\mathcal{C}$, \thref{lem:matmult} shows that the herdability of the system in Equation~\eqref{eq:sys} is determined by walks on $G^{w}$ which have lengths from $1$ to $n$. Further:  
 \begin{lemma}\thlabel{lem:cont} $(\mathcal{C})_{i,(m(d-1)+j)}=\rho_{j\rightarrow i,d}$. \end{lemma}
 \begin{proof}
 	From \thref{lem:matmult},
 	\begin{equation}
 	(A^{d-1}B)_{i,j}={\rho}_{j\rightarrow i,d}. 
 	\end{equation}
 	From the definition of the controllability matrix, the sub-matrix
 	\begin{equation}
 	(\mathcal{C})_{:,m(d-1)+1:md}=A^{d-1}B.
 	\end{equation}

 	The result follows.
 \end{proof}

\section{A Necessary Condition for Complete Herdability}\label{sec:nes}

This section shows how graph structure and system herdability are related by providing a necessary condition for complete herdability of a system known as input connectability. It also explores some examples that show why input connectability is only a necessary condition. 
\begin{definition}A graph is input connectable if \begin{equation} \bigcup_{v_{uj}\in\mathcal{V}_u} \mathscr{R}_{j}=\mathcal{V}_{x}, \end{equation} where $\mathscr{R}_{j}$ is the set of nodes reachable from $v_{uj}$:  $\mathscr{R}_{j}=\{v_{xi}\in\mathcal{V}_{x} \ | \ v_{uj} \rightarrow v_{xi}\}$.\end{definition}

 Input connectability is an important indicator of the ability to control a system and is a necessary condition for structural controllability\cite{Lin1974} and sign controllability \cite{tsatsomeros1998sign}. To see how input connectability impacts herdability requires the following extension of \thref{lem:herdset}.

\begin{lemma}\thlabel{lem:nodeherd}
	A state $i$ is herdable if and only if $\exists j$ such that $$(\mathcal{C})_{i,j}\neq0.$$  
\end{lemma}
\begin{proof}
	($(\mathcal{C})_{i,j}\neq0\Rightarrow$ $i$ Herdable) If $(\mathcal{C})_{i,j}\neq0$ then by appropriate choice of the $j$-th element of a vector $\mathbf{z}$, $(\mathcal{C}\mathbf{z})_{i}=w$ for a positive constant $w$. Then there is a vector $\mathbf{k}\in \mathrm{range}(\mathcal{C})$ with $\mathbf{k}_{i}>0$ and $v_{xi}$ is herdable by \thref{lem:herdset}.
	
	(Herdable $\Rightarrow(\mathcal{C})_{i,j}\neq0$ ) Suppose the contrary. Then by assumption $\forall j \  (\mathcal{C})_{i,j}=0$. Consider making $\mathbf{x}(t)\geq h$ from an initial state $\mathbf{x}(0)=\mathbf{0}_{n}$. As $\forall j \ (\mathcal{C})_{i,j}=0$, it holds that $\forall z\in \mathrm{range}(\mathcal{C}), \ z_{i}=0$ and by \thref{lem:rceq} for any reachable $\mathbf{x}(t) \ \forall t\geq 0$, $\mathbf{x}(t)_{i}=0$ and state $i$ is not herdable. 
\end{proof}

 \begin{theorem}\thlabel{th:incon}
 	If a system is completely herdable, then it is input connectable. 
 \end{theorem}
 \begin{proof}
 	Suppose the system is not input connectable. Then by assumption, there exists a node $v_{xi}$ such that $v_{xi}\notin\bigcup_{j}\mathscr{R}_{j}$ and as such there is no walk from an input to $v_{xi}$.
 	If there is no walk to $v_{xi}$, then $(\mathcal{C})_{i,:}=\mathbf{0}_n$ by \thref{lem:cont} and the node will not be herdable by \thref{lem:nodeherd}. As such, the system is not completely herdable. 
 \end{proof}
 
  Consider the following two examples which illustrate why input connectability is only a necessary condition for herdability. 
  The first example considers the signed dilation, shown in Fig.~\ref{fig:exnec}.
  	 		 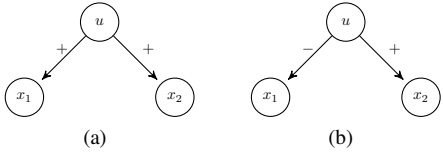
\begin{figure}[h]  	  
  	\centering
  	\subfloat[]{\label{subfig:nex1}	\begin{tikzpicture}[
  		->,
  		>=stealth',
  		shorten >=2pt,
  		auto,
  		node distance=0.5cm,
  		every text node part/.style={align=center},
  		scale=0.25, every node/.style={scale=0.58}
  		]
  		\tikzstyle{every state}=[fill=none,draw=black,text=black]
  		\node[
  		state,
  		] (n0) at(0,0)
  		{$u$};
  		\node[
  		state,
  		] (n1) at(-4,-4)
  		{$x_{1}$};
  		\node[
  		state,
  		] (n2) at(4,-4)
  		{$x_{2}$};
  		\path (n0) edge [->] node [left,above] {$+$}(n1)
  		(n0) edge [->]  node {$+$} (n2);
  		\end{tikzpicture}} \hspace{.5cm}
  	\subfloat[]{\label{subfig:nex2}	\begin{tikzpicture}[
  		->,
  		>=stealth',
  		shorten >=2pt,
  		auto,
  		node distance=0.5cm,
  		every text node part/.style={align=center},
  		scale=0.25, every node/.style={scale=0.58}
  		]
  		\tikzstyle{every state}=[fill=none,draw=black,text=black]
  		\node[
  		state,
  		] (n0) at(0,0)
  		{$u$};
  		\node[
  		state,
  		] (n1) at(-4,-4)
  		{$x_{1}$};
  		\node[
  		state,
  		] (n2) at(4,-4)
  		{$x_{2}$};
  		\path (n0) edge [->] node [left,above] {$-$}(n1)
  		(n0) edge [->]  node {$+$} (n2);
  		\end{tikzpicture}} 
  	\caption{The Signed Dilation: The systems represented by the graph structure in \ref{subfig:nex1} are completely herdable, while \ref{subfig:nex2} shows a graph structure that is never completely herdable.}
  	  		\label{fig:exnec}
  \end{figure}
Fig.~\ref{subfig:nex1} represents systems of the form
\begin{equation}
\dot{\mathbf{x}}=
\begin{bmatrix}
0 & 0 \\
0 & 0  
\end{bmatrix}\mathbf{x} +\begin{bmatrix}
\beta_{1}   \\
\beta_{2}  
\end{bmatrix}\mathbf{u}
\end{equation}
where
$\beta_{1},\beta_{2}>0$,
which gives a controllability matrix:
\begin{equation}
\mathcal{C}=\begin{bmatrix}
\beta_{1} & 0 \\
\beta_{2} & 0 
\end{bmatrix},
\end{equation}
with
\begin{equation}
\mathrm{range}(\mathcal{C})=span\left(\left\{ \begin{bmatrix}
\beta_{1}  \\ 
\beta_{2} 
\end{bmatrix} \right\}\right).
\end{equation}
This system is always completely herdable. On the other hand, Fig.~\ref{subfig:nex2} can be translated to systems of the form: \begin{equation}
\dot{\mathbf{x}}=
\begin{bmatrix}
0 & 0 \\
0 & 0  
\end{bmatrix}\mathbf{x} +\begin{bmatrix}
-\beta_{1}   \\
\beta_{2}  
\end{bmatrix}\mathbf{u}
\end{equation}
where
$\beta_{1},\beta_{2}>0$.
This gives a controllability matrix:
\begin{equation}
\mathcal{C}=\begin{bmatrix}
-\beta_{1} & 0 \\
\beta_{2} & 0 
\end{bmatrix}, 
\end{equation}
with
\begin{equation}
\mathrm{range}(\mathcal{C})=span\left(\left\{ \begin{bmatrix}
-\beta_{1}  \\ 
\beta_{2} 
\end{bmatrix} \right\}\right).
\end{equation}
Here either $v_{x1}$ or $v_{x2}$ can be made larger than all thresholds $h\geq 0$ but not both.  This example illustrates a fundamental trade off when herding signed digraphs, which is that at a given distance from the input either $\nd$ or $\pd$ can be herded but not both. Note that this is a loss of symmetry with respect to the input that causes a loss of herdability. In the language of social networks, it is not possible to simultaneously convince an enemy and a friend.  

  It turns out that Fig.~\ref{subfig:nex1} is an example of a positive system. A system is positive if an element-wise non-negative initial state under element-wise non-negative control input remains element-wise non-negative \cite{Farina2011}. A positive system is excitable if and only if each state variable can be made positive by applying an appropriate nonnegative input to the system initially at rest [$\mathbf{x}(0)=\mathbf{0}_{n}$] \cite{Farina2011}. In the case of a positive system, input connectability is a necessary and sufficient condition for complete herdability. 
 
 \begin{theorem}\thlabel{th:posin}
 	A positive linear system is completely herdable if and only if it is input connectable. 
 \end{theorem}
 \begin{proof}
 	(Sufficiency) By Theorem $8$ of \cite{Farina2011}, an input connectable, positive linear system is excitable. Then there is an element-wise positive vector in the reachable subspace, which is also the range of the controllability matrix by \thref{lem:rceq}. Then by \thref{cor:krange}, the system is completely herdable.
 	 
 	(Necessity) Follows from \thref{th:incon}.
 \end{proof}
 
 The second example that shows why input connectability is only a necessary condition and not a sufficient condition can be seen based on the weighted graph $G^{w}$, specifically the cancellation of walk weights from an input to a state node. It is possible that if a node is included in both $\nd^{j}$ and $\pd^{j}$ that there be a combination of weights such that  ${\rho}_{j\rightarrow i,d}=0$. If the only walks to $v_{xi}$ are of length $d$ then the node $v_{xi}$ is not herdable, as is the case for $v_{x4}$ in Fig.~ \ref{fig:cancel}.

\section{Graph Structure Based Herdability}\label{sec:class}
This section, in contrast to Section \ref{sec:nes}, will use the graph characterization of a system to discuss sufficient conditions for herdability, specifically extending \thref{th:ccond} to consider the underlying graph structure of the system. 

\begin{theorem}\thlabel{th:all}
	If for each $v_{xi}\in\mathcal{V}_{x}$, there exists a distance $d$ and an input $v_{uj}$ such that $v_{xi}\in\mathcal{N}_{d}^{j}\cup\mathcal{P}_{d}^{j}$ and $\mathcal{N}_{d}^{j}=\emptyset \veebar \mathcal{P}_{d}^{j}=\emptyset$ then the system is completely sign herdable. 
\end{theorem}
\begin{proof}
	Consider the herdability of a node $v_{xi}$ which satisfies $v_{xi}\in\mathcal{N}_{d}^{j}\cup\mathcal{P}_{d}^{j}$ and $\mathcal{N}_{d}^{j}=\emptyset \veebar \mathcal{P}_{d}^{j}=\emptyset$ for some $d$ and $v_{uj}$. As $\mathcal{N}_{d}^{j}=\emptyset \veebar \mathcal{P}_{d}^{j}=\emptyset$, it must be that all nodes at distance $d$ from the input $v_{uj}$ are all in either $\mathcal{N}_{d}^{j}$ or $\mathcal{P}_{d}^{j}$. Then by \thref{lem:pdti} and \thref{lem:ndti}, all nonzero elements of $(\mathcal{C})_{:,m(d^{i}-1)+j^{i}}$ have the same sign. As this hold for all $v_{xi}$, the system is completely herdable by \thref{th:ccond}. As the herdability does not depend on the edge weights, this holds for all systems with the same sign pattern and the system is sign herdable. 
\end{proof}

	If a system is not completely herdable, it is still possible to control a subset of the system nodes to enter the set $\mathcal{H}_{d}$. One such case is systems with an underlying graph that is a rooted out-branching. In fact in such a system, it is possible to completely characterize the nodes that can be herded based on the underlying graph structure. A directed graph, $\hat{\mathcal{G}}=(\hat{\mathcal{V}},\hat{\mathcal{E}})$ is a rooted out-branching if it has a root node $v_{i}\in\hat{\mathcal{V}}$ such that for every other node $v_{j}\in\hat{\mathcal{V}}$ there is a single directed walk from $v_{i}$ to $v_{j}$. The case considered here is that of  a single input, input rooted out-branching, which means that every node $v_{xi}\in\hat{\mathcal{V}}_{x}$ has a single in-bound walk from the single input $v_{u}$. The unique walk from $v_{u}$ to $v_{xi}$ in the input-rooted out-branching will be referred to as $\pi_{t}(v_{u},v_{xi})$. Consider the maximum walk length between $v_{u}$ and a state node, which is  
	 $d_{\max}=\max_{v_{xi}\in\hat{\mathcal{V}}_{x}} \mathrm{len}(\pi_{t}(v_{u},v_{xi})).$ 
	  Let $\mathbb{H}_{u}$ be the set of nodes made larger than some lower threshold $h\geq0$ via a signal from the input $v_{u}$. 
	\begin{theorem}\thlabel{th:outtree}
		In an input rooted, out-branching, $\mathbb{H}_{u}$ follows 
		$
		\mathbb{H}_{u}=\bigcup_{d=1}^{d_{\max}} \mathcal{X}_{d},
		$
		 where $\mathcal{X}_{d}\in\{\pd,\nd,\emptyset\}$. 
	\end{theorem}
	\begin{proof}
		 Consider the ability to herd a node $v_{xi}$ and assume that $\mathrm{len}(\pi_{t}(v_{u},v_{xi}))=d_{i}$. As there is only one walk from $v_{u}$ to $v_{xi}$ it holds that $\mathcal{C}_{i,d}=0, \ \forall d\neq d_{i}\in\mathcal{D}.$~Then $\mathcal{C}_{:, d}$ uniquely determines the ability to herd all nodes at distance $d$. If $\alpha_{d}=1$ then $(\mathcal{C}_{:, d}\alpha_{d})_{i}>0, \ \forall i \text{ such that } v_{xi}\in\pd[d_{i}]$ and $\pd[d_{i}]$ is herdable by \thref{lem:herdset}. If $\alpha_{d}=-1$ then  $(\mathcal{C}_{:, d}\alpha_{d})_{j}>0, \ \forall j \text{ such that } v_{xi}\in\nd[d_{i}]$ and $\nd[d_{i}]$ is herdable by \thref{lem:herdset}. Finally if $\alpha_{d}=0$ then  $\mathcal{C}_{:, d}\alpha_{d}=\mathbf{0}_{n}$ and no nodes are herded. 
		
		Construct a vector $\pmb{\alpha}\in\mathbb{R}^{n}$ where $\forall d\in\{1,2,\dots,d_{\text{max}}\} $
		\begin{equation} \pmb{\alpha}_{d}=\begin{cases}1 \text{ \ \ so that } \mathcal{X}_d=\pd,\\ -1 \text{ so that } \mathcal{X}_d=\nd, \\ 0 \text{\ \ \ so that } \mathcal{X}_d=\emptyset,
		\end{cases}
		\end{equation} and where the remaining $n-d_{\max}$ elements are $0$. Then $\mathcal{C}\pmb{\alpha}$ shows the herdability of the set of nodes $\bigcup_{d=1}^{d_{\max}} \mathcal{X}_{d}$. 
	\end{proof}
	\begin{corollary}\thlabel{th:size}
		The maximal collection of nodes, $\mathbb{H}_{u}^{\ast}$, that can be herded in a input rooted out-branching satisfies \begin{equation}
		|\mathbb{H}_{u}^{\ast}|=\sum_{l=1}^{d_\text{max}}\max(|\mathcal{N}_{l}|,|\mathcal{P}_{l}|).
		\end{equation}
		
	\end{corollary}
	\ \\
	
	In the case of an single input, input connectable,  directed out-branching where $\forall d \in \{1,2,\dots,d_{\text{max}}\}, \ \mathcal{N}_{d}=\emptyset \veebar \mathcal{P}_{d}=\emptyset$, \thref{th:size} shows that $|\mathbb{H}_{u}^{\ast}|=n$, i.e. that the system is completely herdable. Fig.~\ref{fig:choose} shows an example of an input rooted, out-branching. 
	\begin{figure}[h]
		\centering
		\begin{tikzpicture}[
		->,
		>=stealth',
		shorten >=2pt,
		auto,
		node distance=0.5cm,
		every text node part/.style={align=center},
		scale=0.25, every node/.style={scale=0.58}
		]
		\tikzstyle{every state}=[fill=none,draw=black,text=black]
		\node[
		state,
		] (n0) at(0,0)
		{$u$};
		\node[
		state,
		] (n1) at(-4,-4)
		{$x_{1}$};
		\node[
		state,
		] (n2) at(4,-4)
		{$x_{2}$};
		\node[
		state,
		] (n3) at(2,-8)
		{$x_{5}$};
		\node[
		state,
		] (n4) at(6,-8)
		{$x_{6}$};
		\node[
		state,
		] (n5) at(-6,-8)
		{$x_{3}$};
		\node[
		state,
		] (n6) at(-2,-8)
		{$x_{4}$};
		\path (n0) edge [->] node [left,above] {$-$}(n1)
		(n0) edge [->]  node {$+$} (n2)
		(n2) edge [->] node[left,above] {$-$} (n3)
		(n2) edge [->]  node {$+$} (n4)
		(n1) edge [->] node[left,above] {$-$} (n5)
		(n1) edge [->]  node {$+$} (n6);
		\end{tikzpicture}
		\caption{An example of an input rooted out-branching}
		\label{fig:choose}
	\end{figure}
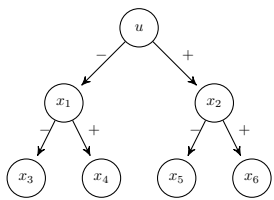
	
	\noindent The graph in Fig.~\ref{fig:choose} represents the following class of systems:
	
	\begin{equation}
	\dot{\mathbf{x}}=\begin{bmatrix}
	0 & 0 & 0 & 0 & 0 & 0 \\ 
	0 & 0 & 0 & 0 & 0 & 0 \\ 
	-\alpha_{1} & 0 & 0 & 0 & 0 & 0 \\ 
	\alpha_{2} & 0 & 0 & 0 & 0 & 0 \\ 
	0 & -\alpha_{3} & 0 & 0 & 0 & 0 \\ 
	0 & \alpha_{4} & 0 & 0 & 0 & 0
	\end{bmatrix}\mathbf{x}+ \begin{bmatrix}
	-\beta_{1}  \\ 
	\beta_{2}  \\ 
	0  \\ 
	0  \\ 
	0  \\ 
	0 
	\end{bmatrix}\mathbf{u} 
	\end{equation}
	where $\alpha_{1},\alpha_{2},\alpha_{3},\alpha_{4},\beta_{1},\beta_{2}>0$. It follows from \thref{th:outtree} (and can be verified by calculation of the controllability matrix $\mathcal{C}$), that the possible sets of herded nodes are $\{1,3,6\},\{1,4,5\},\{2,3,6\},\{2,4,5\}$. 
	
	The result of \thref{th:outtree} is similar in nature to the $k$-walk controllability theory \cite{Gao2014a}. The $k$-walk theory shows that for each $d\in\{1,2,\dots,d_{\text{max}}\}$ a single node at distance $d$ can be controlled. In the graph given in Fig.~\ref{fig:choose}, the possible sets of nodes that can be controlled are $\{1,3\},\{1,4\},\{1,5\},\{1,6\},\{2,3\},\{2,4\}\{2,5\},\{2,6\}.$ The maximal collection of nodes that are controlled in a directed out-branching from input $v_{u}$, $\mathbb{C}_{u}^{\ast}$, satisfies 
	$|\mathbb{C}_{u}^{\ast}|= d_{\text{max}}.$
	In the case of herding a network, \thref{th:size} shows that the maximal collection of nodes, $\mathbb{H}_{u}^{\ast}$, will satisfy
	$d_{\text{max}} \leq|\mathbb{H}_{u}^{\ast}|\leq n.$ The lower bound is achieved if the graph is a directed line graph and the upper bound is achieved if \thref{th:all} is satisfied.  
	In the worst case, the same number of nodes can be herded as can be controlled and depending on the network structure many more nodes can be herded. 
	
\section{Conclusion}\label{sec:conc}
In this paper, we introduce the notion of herdability and present a characterization of the herdability of a system via a condition on the range of the controllability matrix, $\mathcal{C}$. A classification of the underlying system graph allowed for the exploration of some consequences of the condition on the range of $\mathcal{C}$. It was shown that input connectability is a necessary condition for complete herdability and that the sets $\pd^{j}$ and $\nd^{j}$ can be used to characterize both a class of completely herdable systems and the nodes that can be herded in a single input, input rooted out-branching. 
\bibliographystyle{IEEEtran}
	\bibliography{Herdability}

\end{document}